\def\fullversion{1}
\def\draft{0}

\ifnum\fullversion=1
  \newcommand{\apref}[1]{Appendix~\ref{#1}}
\else
  \newcommand{\apref}[1]{the full version of the paper}
\fi

\ifnum\draft=1
  \documentclass[a4paper,envcountsame,draft]{llncs}
\else
  \documentclass[a4paper,envcountsame]{llncs}
\fi

\usepackage[utf8]{inputenc}
\usepackage[T1]{fontenc}
\usepackage[english]{babel}
\usepackage{amsmath, amsfonts, amssymb}
\usepackage{graphicx,tikz}
\usetikzlibrary{decorations.shapes, calc}
\usepackage{newalg}

\newcommand{\Proba}[1]{\mathbb{P}\left[#1\right]}
\newcommand{\QUERY}{\ensuremath{\mathrm{QUERY}}}
\spnewtheorem*{remark*}{Remark}{\bfseries}{\itshape}

\let\doendproof\endproof
\renewcommand\endproof{~\hfill\qed\doendproof}

\title{Graph Reconstruction via Distance Oracles
\texorpdfstring{\ifnum\fullversion=0
    \protect\footnote{Full version available at {\tt http://arxiv.org/??}}
\fi}{}
}
\author{Claire Mathieu\thanks{CNRS. Research supported in part by the French ANR Blanc program under contract ANR-12-BS02-005 (RDAM project) and by the NSF medium AF grant 0964037.}
\and Hang Zhou\thanks{Research supported in part by the French ANR Blanc program under contract ANR-12-BS02-005 (RDAM project).}}
\institute{D\'epartement d'Informatique UMR CNRS 8548, \\
 \'Ecole Normale Sup\'erieure,  Paris, France\\
\email{\tt \{cmathieu,hangzhou\}@di.ens.fr}}
\date{}

\usepackage[pdfpagemode=UseNone]{hyperref}
\hypersetup{
     colorlinks=false,
     pdfborder={0 0 0}
}

\begin{document}
\maketitle
\begin{abstract}
We study the problem of reconstructing a hidden graph given access to a distance oracle.
We design randomized algorithms for the following problems:
reconstruction of a degree bounded graph with query complexity $\tilde{O}(n^{3/2})$;
reconstruction of a degree bounded outerplanar graph with query complexity $\tilde{O}(n)$;
and near-optimal approximate reconstruction of a general graph.
\end{abstract}

\section{Introduction}
Decentralized networks (such as the Internet or sensor networks) raise algorithmic problems different from static, centrally planned networks.
A challenge is  the lack of accurate maps  for the  topology of these networks, due to their dynamical structure and to the lack of centralized control.
How can we achieve an accurate picture of the topology with minimal overhead?
This problem has recently received attention  (see e.g.,~\cite{Beerliova,Chen:2010,dall2006exploring,erlebach2006network}).

For Internet networks, the topology can be investigated at the router and autonomous system (AS) level, where the set of routers (ASs) and their physical connections (peering relations) are the vertices and edges of a graph, respectively. Traditionally, inference of routing topology  has relied on tools such as traceroute and mtrace  to generate path information. However, these tools require cooperation of intermediate nodes or routers to generate messages. Increasingly, routers block traceroute requests due to privacy and security concerns, so inference of topology increasingly relies on delay information rather than on the route itself.  At this level of generality, many problems are provably intractable~\cite{Anandkumar}, thus suggesting the need to study related but simpler questions. In this paper, for simplicity we assume that we have access to every
vertex in the graph, and only the edges are unknown.

\paragraph{The problem.}
Consider the shortest path metric   $\delta(\cdot ,\cdot )$ of a connected, unweighted graph $G=(V,E)$, where $|V|=n$.
In our computational model, we are given the vertex set $V$, and we have access to $\delta$ via a {\em query oracle} \textsc{Query}$(\cdot,\cdot)$ which, upon receiving a query $(u,v)\in V^2$, returns $\delta(u,v)$.  The {\em metric reconstruction problem} is to find the metric $\delta$ on $V$.
The efficiency of a reconstruction algorithm is measured by its \emph{query complexity}, i.e., the number of queries to the oracle. (We focus on query complexity, but our algorithms can also easily be implemented in polynomial time and space).

Note that finding $\delta$ is equivalent to finding every edge in $E$, thus this problem is also called the \emph{graph reconstruction problem}.

\paragraph{Related work.}

Reyzin and Srivastava~\cite{Reyzin2007} showed an $\Omega(n^2)$ lower bound for the graph reconstruction problem on general graphs.
We extend their result to get a lower bound for the graph approximate reconstruction problem.

To reconstruct graphs of bounded degree, we apply some algorithmic ideas previously developed for compact routing~\cite{Thorup:2001:CRS:378580.378581} and ideas for Voronoi cells~\cite{honiden:2009:bgv:1681511.1682426}.

A closely related model in network discovery and verification provides queries which, upon receiving a node $q$, returns the distances from $q$ to all other nodes in the graph~\cite{erlebach2006network}, instead of the distance between a pair of nodes in our model.
The problem of minimizing the number of queries is {\bf NP}-hard and admits an $O(\log n)$-approximation algorithm (see~\cite{erlebach2006network}).
In another model, a query at a node $q$ returns all edges on all shortest paths from $q$ to any other node~\cite{Beerliova}.
Network tomography also proposes statistical models~\cite{Castro04networktomography,Tarissan:2009:EMC:1719850.1719893}.

\paragraph{Our results.}

In Section~\ref{sec:degree bounded}, we consider the reconstruction problem on graphs of bounded degree.
We provide a randomized algorithm to reconstruct such a graph with query complexity $\tilde{O}(n^{3/2})$.
Our algorithm selects a set of nodes (called \emph{centers}) of expected size $\tilde{O}(\sqrt{n})$, so that they separate the graph into $\tilde{O}(\sqrt{n})$ slightly overlapped subgraphs, each of size $O(\sqrt{n})$.
We show that the graph reconstruction problem is reduced to reconstructing every subgraph, which can be done in $O(n)$ queries by exhaustive search inside this subgraph.

In Section~\ref{sec:degree bounded outerplanar}, we consider \emph{outerplanar} graphs of bounded degree.
An \emph{outerplanar} graph is a graph which can be embedded in the plane with all vertices on the exterior face.
Chartrand and Harary~\cite{Chartrand1967} first introduced outerplanar graphs and proved that a graph is outerplanar if and only if it contains no subgraph homeomorphic from $K_4$ or $K_{2,3}$.
Outerplanar graphs have received much attention in the literature because of their simplicity and numerous applications.
In this paper, we show how to reconstruct degree bounded outerplanar graphs with expected query complexity $\tilde{O}(n)$.
The idea is to find the node $x$ which appears most often among all shortest paths (between every pair of nodes), and then partition the graph into components with respect $x$.
We will show that such partition is \emph{$\beta$-balanced} for some constant $\beta<1$, i.e., each resulting component is at most $\beta$ fraction of the graph.
Such partitioning allows us to reconstruct the graph recursively with $O(\log n)$ levels of recursion.
However, it takes too many queries to compute all shortest paths in order to get $x$.
Instead, we consider an approximate version of $x$ by computing a sampling of shortest paths to get the node which is most often visited among all sampling shortest paths.
We will show that the node obtained in this way is able to provide a $\beta$-balanced partition with high probability.
Our algorithm for outerplanar graphs gives an $O(\Delta\cdot n\log^3 n)$ bound which, for a tree (a special case of an outerplanar graph), is only slightly worse than the optimal algorithm for trees with query complexity $O(\Delta \cdot n\log n)$ (see~\cite{hein1989optimal}).
On the other hand, the tree model typically restricts queries to pairs of tree leaves, but we allow queries of any pair of vertices, not just leaves.

In Section~\ref{sec:approx reconstruction}, we consider an approximate version of the metric reconstruction problem for general graphs. The metric $\widehat{\delta}$ is an $f$-approximation of the metric $\delta$ if for every pair of nodes $(u,v)$, $\widehat{\delta}(u,v)\leq \delta(u,v)\leq f\cdot \widehat{\delta}(u,v)$, where $f$ is any sublinear function of $n$.
We give a simple algorithm to compute an $f$-approximation of the metric with expected query complexity $O(n^2(\log n)/f)$.
We show that our algorithm is near-optimal by providing an $\Omega(n^2/f)$ query lower bound.

An open question is whether the $\tilde{O}(n^{3/2})$ bound in Theorem~\ref{thm:boundeddegree} is tight.

\paragraph{Other models.}

The problem of reconstructing an unknown graph by queries that reveal partial information has been studied extensively in many different contexts, independently stemming from a number of applications.

In evolutionary biology, the goal is to reconstruct evolutionary trees, thus the hidden graph has a tree structure.
One may query a pair of species and get in return the distance between them in the (unknown) tree~\cite{waterman1977additive}.
See for example~\cite{hein1989optimal,King,reyzin2007longest}.
In this paper, we assume that our graph is not necessarily a tree, but may have an arbitrary connected topology.

Another graph reconstruction problem is motivated by DNA shotgun sequencing and linkage discovery problem of artificial intelligence~\cite{bouvel2005combinatorial}. In this model we have access to an oracle which receives a subset of vertices and returns the number of edges whose endpoints are both in this subset.
This model has been much studied (e.g.,~\cite{angluin2004learning,Choi,grebinski2000optimal,Reyzin2007}) and an optimal algorithm has been found in~\cite{Mazzawi}. Our model is different since there is no counting.

Geometric reconstruction deals with, for example, reconstructing a curve from a sampling of points~\cite{amenta98thecrust,Dey00reconstructingcurves} or
reconstructing a road network from a given collection of path traces~\cite{Chen:2010}.
In contrast, our problem contains no geometry, so results are incomparable.

\section{Degree Bounded Graphs}
\label{sec:degree bounded}

\begin{theorem}\label{thm:boundeddegree}
Assume that the graph $G$ has bounded degree $\Delta$. Then we have a randomized algorithm for the metric reconstruction problem, with query complexity $O(\Delta^4\cdot n^{3/2}\cdot\log^2 n\cdot \log\log n)$, which is $\tilde{O}(n^{3/2})$ when $\Delta$ is constant.
\end{theorem}
Our reconstruction proceeds in two phases.

In the first phase, we follow the notation from Thorup and Zwick~\cite{Thorup:2001:CRS:378580.378581}: Let $A\subset V$ be a subset of vertices called \emph{centers}. For $v\in V$, let $\delta(A,v)=\min\{\delta(u,v)\mid u\in A\}$ denote the distance from $v$ to the closest node in $A$. For every $w\in V$, let  the \emph{cluster} of $w$ with respect to the set $A$ be defined by $C_w^A=\{v\in V\mid \delta(w,v)<\delta(A,v)\}$. Thus for $w\notin A$, $C_w^A$ is the set of the vertices whose closest neighbor in $A\cup \{ w\}$ is $w$.
Algorithm \textsc{Modified-Center}$(V,s)$, which is randomized, takes as input the vertex set $V$ and a parameter $s\in [1,n]$, and returns a subset $A\subset V$ of vertices such that all clusters $C_w^A$ (for all $w\in V$) are of size at most $6n/s$.
$A$ has expected size at most $2s\log{n}$, thus the expected number of queries is  $O(s\cdot n\cdot \log^2 n \cdot \log \log n)$. This algorithm applies, in a different context, ideas from~\cite{Thorup:2001:CRS:378580.378581}, except that we use sampling to compute an estimate of $|C_w^A|$.

In the second phase, Algorithm \textsc{Local-Reconstruction}$(V,A)$ takes as input the vertex set $V$ and the set $A$ computed by \textsc{Modified-Center}$(V,s)$, and returns the edge set of $G$.
It partitions the graph into slightly overlapped components according to the centers in $A$, and proceeds by exhaustive search within each component.
Inspired by the Voronoi diagram partitioning in~\cite{honiden:2009:bgv:1681511.1682426}, we show that these components together cover every edge of the graph.
The expected query complexity in this phase is  $O(s\log{n}(n+\Delta^4(n/s)^2))$.

Letting $s=\sqrt{n}$, the expected total number of queries in the two phases is $O(\Delta^4\cdot n^{3/2}\cdot\log^2 n\cdot \log\log n)$.

We use the notation \textsc{Query}$(A,v)$ to mean \textsc{Query}$(a,v)$ for every $a\in A$,
and the notation \textsc{Query}$(A,B)$ to mean \textsc{Query}$(a,b)$ for every $a\in A$ and $b\in B$.

\vspace{3mm}
\begin{algorithm}{Modified-Center}{V,s}
A\=\emptyset ,\; W\=V\\
T\=K \cdot\log n\cdot \log{\log n} \text{ ($K=O(1)$ to be defined later)} \\
\begin{WHILE}{W\neq \emptyset}
 A'\= \text{Random subset of $W$ s.t.\ every node has prob. ${s}/{|W|}$}\\
 \CALL{Query}(A',V)\\
 A\=A \cup A'\\
\begin{FOR}{w\in W}
    X\=\text{Random multi-subset of} V \text{with} s\cdot T \text{elements}\\
    \CALL{Query}(X,w)\\
    \text{Let }\widehat{C_w^A}\=|X\cap C_w^A|\cdot n/|X|
\end{FOR}\\
W\=\{w\in W: \widehat{C_w^A}\geq 5n/s\}
\end{WHILE}\\
\RETURN A
\end{algorithm}

\begin{algorithm}{Local-Reconstruction}{V,A}
 E\=\emptyset\\
 \begin{FOR}{a\in A}
  B_a\=\{v\in V \mid \delta(v,a)\leq 2\}\\
  \CALL{Query}(B_a,V)\\
  D_a\=B_a\\
  \begin{FOR}{b\in B_a}
   D_a\=D_a\cup\{v\in V\mid\delta(b,v)<\delta(A,v)\}
 \end{FOR}\\
 \CALL{Query}(D_a,D_a)\\
  E\=E\cup \{(d_1,d_2)\in D_a\times D_a: \delta(d_1,d_2)=1\}
\end{FOR}\\
\RETURN E
\end{algorithm}

Figure~\ref{fig:voronoi} gives an illustration of Algorithm \textsc{Local-Reconstruction}($V,A$). Vertices $a_1,\dots,a_5$ are centers in $A$ and define subsets $D_{a_1},\dots,D_{a_5}$ which overlap slightly.
We will show in Lemma~\ref{lemma:shift} that the subsets $D_a$ (for all $a\in A$) together cover every edge in $E$. Thus the local reconstruction over every $D_a$ (for $a\in A$) is sufficient to reconstruct the graph.

\begin{figure}[ht]
\centering
\scalebox{0.5}{
\ifnum\fullversion=1
    \includegraphics{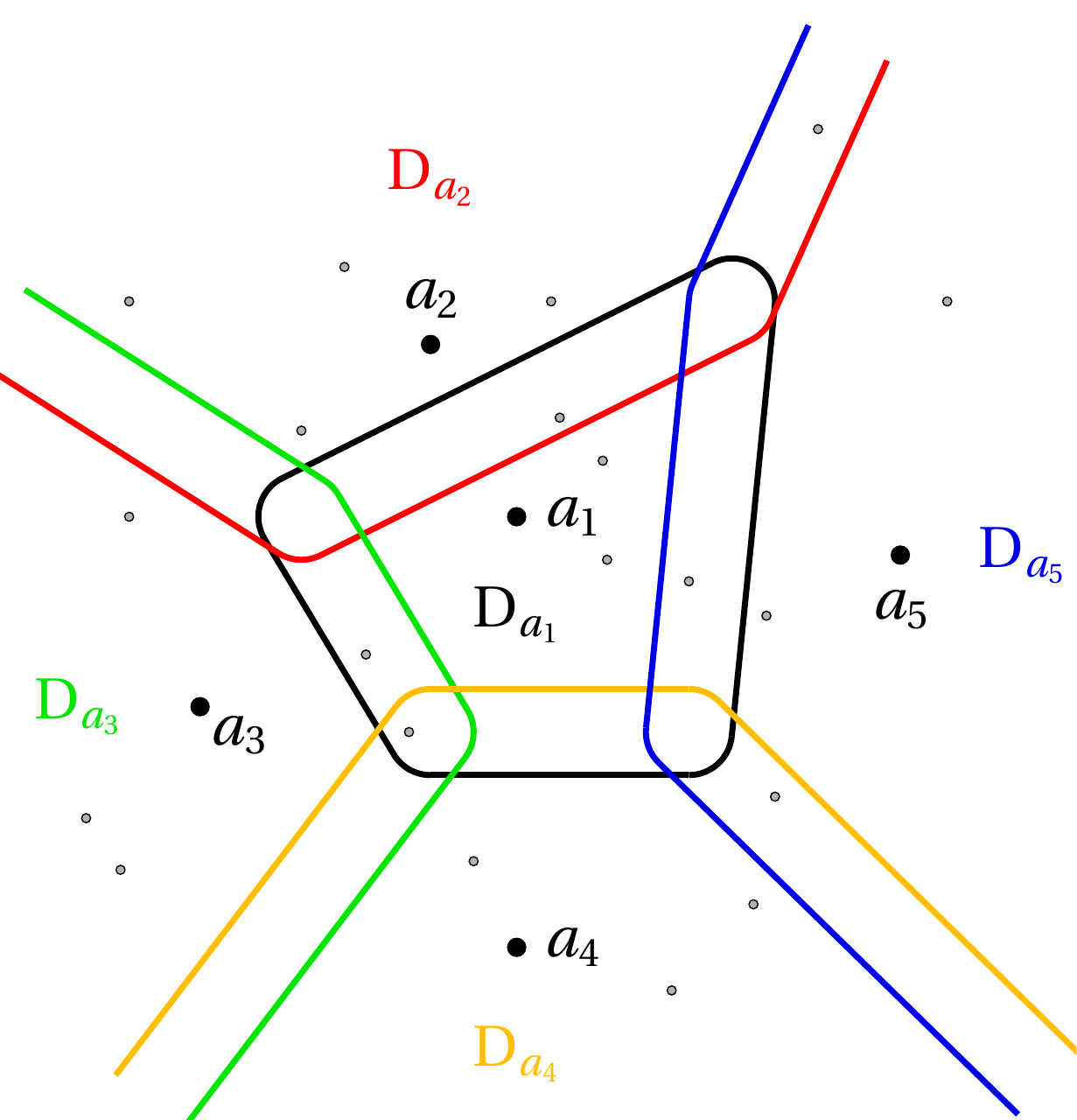}
\else
    \includegraphics{figure_black.pdf}
\fi
}
\caption{Partition by centers}
\label{fig:voronoi}
\end{figure}

Theorem~\ref{thm:boundeddegree} follows from Lemma~\ref{lemma:balanced separation} and~\ref{lemma:shift}.
\begin{lemma}
\label{lemma:balanced separation}
With probability at least $1/(4e)$, the \textsc{Modified-Center}$(V,s)$ algorithm takes  $O(s\cdot n\cdot \log^2 n \cdot \log \log n)$ queries and returns a set $A$ of size at most $4s\log{n}$ such that $|C_w^A|\leq 6n/s$ for every $w\in V$.
\end{lemma}

\begin{remark*}
The difference between our algorithm and algorithm \textsc{Center}($G,s$) in~\cite{Thorup:2001:CRS:378580.378581} is that, \textsc{Center}($G,s$) eliminates $w\in W$ when $|C_w^A|<4n/s$, by calculating $|C_w^A|$ exactly, which needs $n$ queries in our model; while our algorithm gives an estimation of $|C_w^A|$ using $O(s\cdot \log n\cdot \log\log n)$ queries, so that with high probability, it eliminates $w\in W$ when $C_w^A<4n/s$ and it does not eliminate $w\in W$ when $C_w^A>6n/s$.
\end{remark*}

\begin{proof}
Fix $A$ and $w$ and let $Y_w=|X\cap C_w^A|= |\{x\in X\mid \delta(x,w)<\delta(x,A)\} |$. The expected value of $Y_w$ is $|C_w^A|\cdot |X|/n$. Since $X$ is random, by standard Chernoff bounds there is a constant $K$ such that, for any node $w$,
\[
  \begin{cases}
    \Proba{Y_w>5T} > 1-1/(4n\log{n}), &\text{ if } C_w^A>6n/s \text{ (and thus $E[Y_w]>6T$)}\\
    \Proba{Y_w<5T} > 1-1/(4n\log{n}), &\text{ if } C_w^A<4n/s \text{ (and thus $E[Y_w]<4T$)}.
  \end{cases}
\]
Let $\widehat{C_w^A}=Y_w\cdot n/|X|$, where $|X|=s\cdot T$.
When the number of nodes $w$ in estimation is at most $4n\log{n}$, with probability at least $(1-1/(4n\log{n}))^{4n\log{n}}\sim 1/e$, we have:
  \begin{equation}
  \label{chernoff}
  \begin{cases}
    \widehat{C_w^A}> 5n/s, &\text{ if } C_w^A>6n/s \\
    \widehat{C_w^A}< 5n/s, &\text{ if } C_w^A<4n/s
  \end{cases},\text{ for every $w$ in estimation. }
  \end{equation}
We assume that $n$ is large enough  that this probability is at least $1/(2e)$.

Using the same proof as that of Theorem~3.1 in~\cite{Thorup:2001:CRS:378580.378581}, we can prove that under  condition~(\ref{chernoff}), algorithm \textsc{Modified-Center}$(V,s)$ executes an expected number of at most $2\log n$ iterations of the while loop and returns a set $A$ of expected size at most $2s\log{n}$ such that $|C_w^A|\leq6n/s$ for every $w\in V$. Thus with probability at least $1/2$, the algorithm executes at most $4\log n$ iterations of the while loop and the set $A$ is of size at most $4s\log{n}$. The number of queries is $O(s\cdot n\cdot \log^2 n \cdot \log \log n)$ in this case, since every iteration takes $O(s\cdot n\cdot \log n\cdot \log\log n)$ queries.
So the lemma follows.
\end{proof}

\begin{lemma}\label{lemma:shift}
Under the conditions that  $|A|\leq 4s\log{n}$ and $|C_w^A|\leq 6n/s$ for every $w\in V$, Algorithm \textsc{Local-Reconstruction}$(V,A)$ finds all edges in the graph using $O(s\log{n}(n+\Delta^4(n/s)^2))$ queries.
\end{lemma}

\begin{proof}
Let $D_a=B_a\bigcup_{b\in B_a} C_b^A$. We will prove that for every edge $(u,v)$ in $E$, there is some $a\in A$, such that  $u$ and $v$ are both in $D_a$. Thus the algorithm is correct: it finds all edges in $E$.

Consider $(u,v)\in E$. Without loss of generality, we assume $\delta(A,u)\leq \delta(A,v)$.
Let $a\in A$ be such that $\delta(a,u)=\delta(A,u)$.
We will show that $u$ and $v$ are both in $D_a$.
When $\delta(a,u)\leq 1$, $u$ and $v$ are both in $B_a\subseteq D_a$.
So we consider only $\delta(a,u)\geq 2$.
Take $b$ to be the node, in any of the shortest paths from $a$ to $u$, such that $\delta(a,b)=2$.
Then $\delta(b,u)=\delta(a,u)-2$ and $\delta(b,v)\leq \delta(b,u)+\delta(u,v)=\delta(a,u)-1$ by the triangle inequality.
Using $\delta(a,u)=\delta(A,u)\leq \delta(A,v)$, we have $\delta(b,u)< \delta(A,u)$ and $\delta(b,v)<\delta(A,v)$. So $u$ and $v$ are both in $C_b^A$, which is a subset of $D_a$ since $b\in B_a$.

Because every $D_a$ (for $a\in A$) has size at most $\Delta^2\cdot 6n/s$, the total query complexity is $O(s\log{n}(n+\Delta^4(n/s)^2))$.
\end{proof}

\section{Degree Bounded Outerplanar Graphs}
\label{sec:degree bounded outerplanar}
In this section, we consider the connected graph $G=(V,E)$ to be \emph{outerplanar}~\cite{Chartrand1967} and of bounded degree $\Delta$.
We show how to reconstruct such a graph with expected query complexity $\tilde{O}(n)$.
Generally speaking, we partition the graph into balanced-sized subgraphs and recursively reconstruct these subgraphs.

\subsection{Self-contained Subsets, Polygons and Partitions}
Before giving details of the algorithm, we first need some new notions.

\begin{definition}
The subset $U\subseteq V$ is said to be \emph{self-contained}, if for every $(x,y)\in U\times U$, any shortest path in $G$ between $x$ and $y$ contains nodes only in $U$.
\end{definition}

For every subset $U\subseteq V$, note $G[U]$ to be the subgraph \emph{induced} by $U$, i.e., $G[U]$ has exactly the edges over $U$ in the graph.
It is easy to see that for every self-contained subset $U$, $G[U]$ is outerplanar and connected; and that the intersection of several self-contained subsets is again self-contained.

\begin{definition}
We say that the $k$-tuple $(x_1,\dots,x_k)\in V^k$ (where $k\geq 3$) forms a \emph{polygon} if $G[\{x_1,\dots,x_k\}]$ has exactly $k$ edges:  $(x_1,x_2),(x_2,x_3),\dots,(x_k,x_1)$.
\end{definition}

\begin{definition}
Let $U$ be a self-contained subset of $V$ and let $U_1,\dots,U_\eta$ be subsets of $U$.
We say that $\{U_1,\dots,U_\eta\}$ is a \emph{partition} of $U$ if every $U_i$ is self-contained, and for every edge $(x,y)$ in $G[U]$, there exists some $U_i$ ($1\leq i\leq \eta$) such that $x$ and $y$ are both in $U_i$.
Let $\beta<1$ be some constant.
The partition $\{U_1,\dots,U_\eta\}$ of $U$ is said to be \emph{$\beta$-balanced} if every $U_i$ is of size at most $\beta |U|$.
\end{definition}

Given any partition of $U$, the reconstruction problem over $U$ can be reduced to the independent reconstruction over every $U_i$ ($1\leq i\leq \eta$).

Let $U$ be a self-contained subset of $V$.
For every vertex $v\in U$, its removal would separate $U$ into $n_v$ $(n_v\geq 1)$ connected components.
For every $i\in [1,n_v]$, let $S_{v,i}^*$ be the set of nodes in the $i^{\rm th}$ component and let $S_{v,i}=S_{v,i}^*\cup \{v\}$.
We say that $\{S_{v,1},\dots,S_{v,n_v}\}$ is the \emph{partition of $U$ by the node $v$}.

\subsection{Balanced-Partition Algorithm}
Let us now introduce the main algorithm \textsc{Balanced-Partition}$(U)$, which takes as input a self-contained subset $U\subseteq V$ with $|U|\geq 10$ and returns a $\beta$-balanced partition of $U$, for some constant $\beta\in(0.7,1)$.
The algorithm takes a sampling of $2\omega$ nodes $(a_1,\dots,a_\omega,b_1,\dots,b_\omega)$, where $\omega=C\cdot \log |U|$ for some constant $C>1$, and tries to find a $\beta$-balanced partition of $U$ under this sampling.
It stops if it finds such a partition, and repeatedly tries another sampling otherwise.
Below is the general framework of our algorithm.
The details of the algorithmic implementation are given in \apref{appendix:algo-implement}, where we give the constants $C$ and $\beta$.

\begin{enumerate}
\item Take a sampling of $2\omega$ nodes $(a_1,\dots,a_\omega,b_1,\dots,b_\omega)$.
      For every $i\in[1,\omega]$, compute a shortest path between $a_i$ and $b_i$.
      Let $x$ be some node with the most occurrences in the $\omega$ paths above.
\item Partition $U$ into $S_{x,1},\dots,S_{x,n_x}$ by the node $x$.
      If all these sets have size at most $\beta |U|$, return $\{S_{x,1},\dots,S_{x,n_x}\}$;
      otherwise let $D=S_{x,k}$ be the largest set among them and let $V_0=U\backslash S_{x,k}^*$.
\item In the set $D$, compute the neighbors of $x$ in order: $y_1,\dots,y_{\lambda}$, where $\lambda\leq \Delta$.
    If $\lambda=1$, go to Step 1.
\item For every $i\in [1,\lambda]$, partition $U$ into $S_{y_i,1},\dots,S_{y_i,n_{y_i}}$ by $y_i$. Let $S_{y_i,k_i}$  be the subset containing $x$ and let $V_i=U\backslash S_{y_i,k_i}^*$ (see Figure~\ref{fig:partition by neighbors}). If $|V_i|>\beta |U|$, go to Step 1.
\item Let $T=D\cap S_{y_1,k_1}\cap\cdots\cap S_{y_{\lambda},k_\lambda}$.
      Separate $T$ into subsets $T_1$,\dots,$T_{\lambda-1}$ as in Figure~\ref{fig:partition by neighbors}.
      If every $T_i$ has at most $\beta |U|$ nodes, return $\{T_1,\dots,T_{\lambda-1},V_0,\dots,V_\lambda\}$.
\item Let $T_j$ be the set with more than $\beta |U|$ nodes. Find the unique polygon $(q_1,\dots,q_l)$ in $T_j$ that goes by nodes $x$, $y_j$ and $y_{j+1}$.
\item For every $i\in[1,l]$, partition $U$ into $S_{q_i,1},\dots,S_{q_i,n_{q_i}}$ by $q_i$. Let $S_{q_i,m_i}$ be the subset containing the polygon above and let $W_i=U\backslash S_{q_i,m_i}^*$ (see Figure~\ref{fig:partition by polygon}). If $|W_i|>\beta |U|$, go to Step 1.
\item Let $R=S_{q_1,m_1}\cap\cdots\cap S_{q_l,m_l}$. Separate $R$ into subsets $R_1,\dots,R_l$ as in Figure~\ref{fig:partition by polygon}. If some $R_i$ has more than $\beta |U|$ nodes, go to Step 1; else return $\{R_1,\dots,R_l, W_1,\dots,W_l\}$.
\end{enumerate}

\tikzstyle{vertex}=[font=\small]
\begin{figure}
\begin{minipage}[t]{0.48\textwidth}
\scalebox{0.8}{
\begin{tikzpicture}
\node[vertex] (x) at (90:2) {$x$};
\node[vertex] (y1) at (38.5714:2) {$y_1$};
\node[vertex] (y2) at (-12.8572:2) {};
\node[vertex] (y3) at (-64.2858:2) {$y_i$};
\node[vertex] (y4) at (-115.7144:2) {$y_{i+1}$};
\node[vertex] (y5) at (-167.1430:2) {};
\node[vertex] (y6) at (-218.5716:2) {$y_\lambda$};

\foreach \from/\to in {y1/y2,y2/y3,y3/y4/,y4/y5,y5/y6}
\path (\from) edge [dashed, bend left=50,looseness=1] (\to);

\foreach \to in {y1,y2,y3,y4,y5,y6}
\path (x) edge (\to);

\draw [rounded corners=15,dashed] (x) -- (75:3.5) -- (105:3.5) -- (x);
\draw [rounded corners=15,dashed] (y1) -- (23.571400:3.5) -- (53.571400:3.5) -- (y1);
\draw [rounded corners=15,dashed] (y2) -- (-27.857200:3.5) -- (2.142800:3.5) -- (y2);
\draw [rounded corners=15,dashed] (y3) -- (-79.285800:3.5) -- (-49.285800:3.5) -- (y3);
\draw [rounded corners=15,dashed] (y4) -- (-130.714400:3.5) -- (-100.714400:3.5) -- (y4);
\draw [rounded corners=15,dashed] (y5) -- (-182.143000:3.5) -- (-152.143000:3.5) -- (y5);
\draw [rounded corners=15,dashed] (y6) -- (-233.571600:3.5) -- (-203.571600:3.5) -- (y6);
\node[vertex] at (90:3){$V_0$};
\node[vertex] at (38:3) {$V_1$};
\node[vertex] at (-12:3) {};
\node[vertex] at (-64:3) {$V_i$};
\node[vertex] at (-115:3) {$V_{i+1}$};
\node[vertex] at (-167:3) {};
\node[vertex] at (-218:3) {$V_\lambda$};
\node[vertex]  at ($(y1)!0.4!(y2)$) {$T_1$};
\node[vertex]  at ($(y6)!0.4!(y5)$) {$T_{\lambda-1}$};
\node[vertex]  at ($(y3)!0.5!(y4)!0.05!(x)$) {$T_i$};
\end{tikzpicture}
}
\caption{Partition by neighbors}
\label{fig:partition by neighbors}
\end{minipage}
\hfill
\begin{minipage}[t]{0.48\textwidth}
\scalebox{0.8}{
\begin{tikzpicture}
\node[vertex] (q1) at (90:2) {$q_1$};
\node[vertex] (q2) at (38.5714:2) {$q_2$};
\node[vertex] (q3) at (-12.8572:2) {};
\node[vertex] (q4) at (-64.2858:2) {$q_i$};
\node[vertex] (q5) at (-115.7144:2) {$q_{i+1}$};
\node[vertex] (q6) at (-167.1430:2) {};
\node[vertex] (q7) at (-218.5716:2) {$q_l$};
\foreach \from/\to in {q1/q2,q2/q3,q3/q4/,q4/q5,q5/q6,q6/q7,q7/q1}{
\path (\from) edge [dashed, bend left=60,looseness=1] (\to);
\path (\from) edge (\to);
}
\path (q1) edge node[above,vertex]{$R_1$} (q2);
\path (q1) edge node[above,vertex]{$R_l$} (q7);
\path (q4) edge node[below,vertex]{$R_i$} (q5);
\draw [rounded corners=15,dashed] (q1) -- (75:3.5) -- (105:3.5) -- (q1);
\draw [rounded corners=15,dashed] (q2) -- (23.571400:3.5) -- (53.571400:3.5) -- (q2);
\draw [rounded corners=15,dashed] (q3) -- (-27.857200:3.5) -- (2.142800:3.5) -- (q3);
\draw [rounded corners=15,dashed] (q4) -- (-79.285800:3.5) -- (-49.285800:3.5) -- (q4);
\draw [rounded corners=15,dashed] (q5) -- (-130.714400:3.5) -- (-100.714400:3.5) -- (q5);
\draw [rounded corners=15,dashed] (q6) -- (-182.143000:3.5) -- (-152.143000:3.5) -- (q6);
\draw [rounded corners=15,dashed] (q7) -- (-233.571600:3.5) -- (-203.571600:3.5) -- (q7);
\node[vertex] at (90:3){$W_1$};
\node[vertex] at (38:3) {$W_2$};
\node[vertex] at (-12:3) {};
\node[vertex] at (-64:3) {$W_i$};
\node[vertex] at (-115:3) {$W_{i+1}$};
\node[vertex] at (-167:3) {};
\node[vertex] at (-218:3) {$W_l$};
\end{tikzpicture}
}
\caption{Partition by polygon}
\label{fig:partition by polygon}
\end{minipage}
\hfill
\end{figure}

In \apref{appendix:algo-implement}, we give formal definitions and algorithms for subproblems: shortest path between two nodes; partition $U$ by a given node; obtain the neighbors of $x$ in order; partitions $U$ with respect to an edge;  and find the unique polygon that goes by nodes $x$, $y_j$ and $y_{j+1}$.
Finally, we give an improved implementation of partitioning $U$ by a polygon (Steps 7 - 8).
All these algorithms use $O(\Delta\cdot|U|\log^2 |U|)$ queries. It is easy to see that the algorithm \textsc{Balanced-Partition}$(U)$ always stops with a $\beta$-balanced partition of $U$.

\subsection{From Balanced Partitioning to Graph Reconstruction}

Let us show how to reconstruct the graph using \textsc{Balanced-Partition}$(U)$ assuming the following proposition, which will be proved in Section~\ref{sec:complexity analysis}.
\begin{proposition}
\label{balanced partition}
For any self-contained subset $U\subseteq V$ with $|U|\geq 10$, the randomized algorithm \textsc{Balanced-Partition}$(U)$ returns a $\beta$-balanced partition of $U$ with query complexity $O(\Delta\cdot|U|\log^2 |U|)$.
\end{proposition}

Based on the algorithm \textsc{Balanced-Partition}$(U)$, we reconstruct the graph recursively: we partition the vertex set $V$ into self-contained subsets $V_1,\dots,V_k$ such that every $V_i$ has size $\leq \beta n$; for every $V_i$, if $|V_i|<10$, we reconstruct $G[V_i]$ using at most $9^2$ queries; otherwise we partition $V_i$ into self-contained subsets of size at most $\beta|V_i|\leq \beta^2 n$, and continue with these subsets, etc.
Thus the number of levels $L$ of the recursion is $O(\log n)$.

Every time \textsc{Balanced-Partition}$(U)$ returns a partition $\{U_1, \dots,U_k\}$, we always have $|U_1|+\cdots+|U_k|\leq |U|+2(k-1)$.
For every $1\leq i\leq L$, let $U_{i,1},\dots,U_{i,M_i}$ be all sets on the $i^{\rm th}$ level of the recursion.
We then have $|U_{i,1}|+\cdots+|U_{i,M_i}|\leq 3n$.
Thus the total query complexity on every level is $O(\Delta\cdot n\log^2 n)$ by Proposition~\ref{balanced partition}.
So we have the following theorem.

\begin{theorem}
\label{thm:outerplanar}
Assume that the outerplanar graph $G$ has bounded degree $\Delta$.
We have a randomized algorithm for the metric reconstruction problem with query complexity $O(\Delta\cdot n\log^3 n)$, which is $\tilde{O}(n)$ when $\Delta$ is constant.
\end{theorem}

\subsection{Complexity Analysis of the \textsc{Balanced-Partition} Algorithm}
\label{sec:complexity analysis}
Now let us prove Proposition~\ref{balanced partition}.
Since the query complexity to try every sampling is $O(\Delta\cdot|U|\log^2 |U|)$, we only need to prove, as in the following proposition, that for every sampling, the algorithm finds a $\beta$-balanced partition with high probability. This guarantees that the average number of samplings is a constant, which gives the $O(\Delta\cdot|U|\log^2 |U|)$ query complexity in Proposition~\ref{balanced partition}.

\begin{proposition}
\label{failure probability}
In the algorithm \textsc{Balanced-Partition}$(U)$, every sampling of $(a_1,\dots,a_\omega,b_1,\dots,b_\omega)$ gives a $\beta$-balanced partition with probability at least $2/3$.
\end{proposition}

To prove Proposition~\ref{failure probability}, we need Lemmas~\ref{lemma:self-contained},~\ref{most popular node} and~\ref{sample accuracy}, whose proofs are in \apref{appendix:proofs}.

\begin{lemma}\label{lemma:self-contained}
Let $(a_1,\dots,a_\omega,b_1,\dots,b_\omega)$ be any sampling during the algorithm \textsc{Balanced-Partition}$(U)$.
Let $x$ be the node computed from this sampling in Step 1.
We say that a set $S$ is a \emph{$\beta$-bad set}, if it is a self-contained subset of $U$ such that $x\notin S$ and $|S|\geq \beta |U|$ for some constant $\beta$.
Then $x$ does not lead to a  $\beta$-balanced partition of $U$ only when there exists some $\beta$-bad set.
\end{lemma}

For any node $u\in U$, define $p_u$ to be be the probability  that $u$ is in at least one of the shortest paths between two nodes $a$ and $b$, where $a$ and $b$ are chosen uniformly and independently at random from $U$.

\begin{lemma}
\label{most popular node}
There exists some constant $\alpha\in (0,1)$, s.t.\ in every outerplanar graph of bounded degree, there is a node $z$ with $p_z\geq\alpha$.
\end{lemma}

\begin{lemma}
\label{sample accuracy}
Let $\omega=C\cdot\log |U|$ (for some constant $C$ to be chosen in the proof). Take a sample of $2\omega$ nodes uniformly and independently at random from $U$. Let them be $a_1,\dots, a_\omega, b_1, \dots, b_\omega$. For every $v\in U$, let $\widehat{p}_v$ be the percentage of pairs $(a_i,b_i)_{1\leq i\leq \omega}$ such that $v$ is in some shortest path between $a_i$ and $b_i$. Let $x$ be some node in $U$ with the largest $\widehat{p}_x$. Then with probability at least $2/3$, we have $p_x>\alpha/2$, where $\alpha>0$ is the constant in Lemma~\ref{most popular node}.
\end{lemma}

Now we will prove Proposition~\ref{failure probability}.
By Lemma~\ref{lemma:self-contained}, we only need to bound the probability of existence of $\beta$-bad set.
Let $C$ be the constant chosen in Lemma~\ref{sample accuracy}.
Let $x$ be the node computed from the sampling $(a_1,\dots,a_\omega,b_1,\dots,b_\omega)$ in Step 1 of Algorithm \textsc{Balanced-Partition}$(U)$.
Take $\beta=\sqrt{1-\alpha/2}$, where the constant $\alpha\in (0,1)$ is provided by Lemma~\ref{most popular node}. Then $\beta\in (0.7,1)$.
Suppose there exists a $\beta$-bad set $S$.
For every $(a,b)\in S\times S$, any shortest path between $a$ and $b$ cannot go by $x$, since $S$ is self-contained. So $p_x\leq 1-(|S|/|U|)^2\leq 1-\beta^2= \alpha/2$.
By Lemma~\ref{sample accuracy}, the probability that $p_x\leq \alpha/2$ is at most $1/3$.
So the probability of existence of $\beta$-bad set is at most $1/3$.
Thus we complete the proof.

\section{Approximate Reconstruction on General Graphs}
\label{sec:approx reconstruction}
In this section, we study the approximate version of the metric reconstruction problem.
We first give an algorithm for the approximate reconstruction, and then show that this algorithm is near-optimal by providing a query lower bound which coincides with its query complexity up to a logarithmic factor.

\begin{definition}
Let $f$ be any sublinear function of $n$. An \emph{$f$-approximation} $\widehat{\delta}$ of the metric $\delta$ is such that, for every $(u,v)\in V^2$, $\widehat{\delta}(u,v)\leq \delta(u,v)\leq f\cdot\widehat{\delta}(u,v).$
\end{definition}

The following algorithm \textsc{Approx-Reconstruction}$(V)$ receives the vertex set $V$ and samples an expected number of $O(n(\log n)/f)$ nodes. For every sampled node $u$, it makes all queries related to $u$ and provides an estimate $\widehat{\delta}(v,w)$ for every $v$ within distance $f/2$ from $u$ and every $w\in V\backslash\{v\}$.

\vspace{3mm}
\begin{algorithm}{Approx-Reconstruction}{V}
 \begin{WHILE}{\widehat{\delta} \text{is not defined on every pair of nodes}}
  u\=\text{a node chosen from $V$ uniformly at random}\\
  \begin{FOR}{\hbox{every } v\in V}
  \CALL{Query}(u,v) \hbox{ and let }\widehat\delta(u,v)\=\delta(u,v).
  \end{FOR}\\
  S_u\=\{v: \delta(u,v)< f/2\}\\
  \begin{FOR}{v\in S_u\setminus\{ u\}}
  	\begin{FOR}{w\in S_u\backslash\{v\}}
		\widehat{\delta}(v,w)\=1
		\end{FOR}\\
	\begin{FOR}{w\notin S_u}
		\widehat{\delta}(v,w)\=  \delta(u,w)-\delta(u,v)
		\end{FOR}
 \end{FOR}
 \end{WHILE}\\
 \RETURN \widehat{\delta}
\end{algorithm}

\begin{theorem}
The randomized algorithm \textsc{Approx-Reconstruction}$(V)$ computes an $f$-approximation $\widehat \delta$ of the metric $\delta$ using $O(n^2 (\log n)/f)$ queries.
\end{theorem}

\begin{proof}
\renewcommand\endproof{\doendproof}
First we prove that for every $(v,w)$, we have $\widehat{\delta}(v,w)\leq \delta(v,w)\leq f\cdot \widehat{\delta}(v,w)$.
There are two cases:

{Case 1}: $w\in S_u\backslash\{v\}$ (line 7). Then $\widehat{\delta}(v,w)=1\leq \delta(v,w)\leq \delta(u,v)+\delta(u,w)<(f/2)+(f/2)= f=f \cdot \widehat{\delta}(v,w)$, because $v$ and $w$ are in $S_u$.

{Case 2}: $w\notin S_u$ (line 9). On the one hand, by the triangular inequality,
$\delta(v,w)\geq \delta(w,u)-\delta(v,u)= \widehat\delta(v,w)$.
On the other hand, by the triangular inequality, $\delta(w,v)  \leq  (\delta(u,w)-\delta(u,v)) + 2\delta(u,v)$. The first tem is $ \widehat\delta(v,w)$. The second term, by definition of $S_u$, is at most $f-1$. Since $v\in S_u$ and $w\notin S_u$, we have $\delta(u,w)-\delta(v,w)\geq 1$, so the second term can be bounded by $f-1\leq (f-1)\cdot\widehat\delta(v,w)$. Adding completes the proof of the upper bound.

Next, we analyze the query complexity of the algorithm. Since $G$ is connected, for every node $v$ there are at least $f/2$ points  $u$ such that $v\in S_u$. Let $X$ denote all samples during the algorithm. The number of queries is $n|X|$, and its expectation is $n\sum_t \Pr [|X|>t]$. Let $X_t$ denote the first $t$ samples chosen. We have:
$$\Pr [|X|>t]=\Pr[\exists v, \forall u\in X_t, v\notin S_u ]\leq
\left\{ \begin{array}{l}
1  \hbox{ \qquad \qquad \qquad \qquad if } t<2n(\ln n )/f  \\
 \sum_v \Pr[\forall u\in X_t, v\notin S_u]  \hbox{ \quad otherwise.}
\end{array} \right. $$
By independence, $\Pr[\forall u\in X_t, v\notin S_u]\leq (1-(f/2)/n)^t\leq e^{-tf/(2n)}$. Thus
\[E [\# \hbox{queries}]\leq n \frac{2n\ln n}{f}+n^2 \frac{(1/n)}{f/(4n)} = O(n^2(\log n)/f). \tag*{\qed}\]
\end{proof}

On the lower bound side, Reyzin and Srivastava proved a tight $\Omega(n^2)$ bound for the exact reconstruction problem, as in the following proposition.
\begin{proposition}\label{prop:exact-lower-bound}\cite{Reyzin2007}
\label{exact reconstruction lower bound}
Any deterministic or randomized algorithm for the exact graph reconstruction problem requires $\Omega(n^2)$ queries.
\end{proposition}
We extend the proof of Proposition~\ref{prop:exact-lower-bound} to get a lower bound for approximate reconstruction as in Theorem~\ref{prop:approx-lower-bound}, whose proof is in \apref{appendix:approx-lower-bound}.

\begin{theorem}\label{prop:approx-lower-bound}
Any deterministic or randomized approximation algorithm requires $\Omega(n^2/f)$ queries to compute an $f$-approximation of the graph metric.
\end{theorem}

\section*{Acknowledgments}
We would like to thank Fabrice Ben Hamouda for his helpful comments.

\bibliographystyle{abbrv}

\bibliography{reference}
%
%
%

\ifnum\fullversion=1

\appendix

\clearpage
\section{Implementation of \textsc{Balanced-Partition} Algorithm}
\begin{definition}
Define a \emph{boundary cycle} $(b_1, \dots ,b_t)$ to be a cycle of vertices along the whole boundary of the unbounded face, where $b_i$ and $b_{i+1}$ are connected by an edge for every $i\in [1,t]$ with $b_{t+1}=b_1$.
\end{definition}
Notice that a boundary cycle may contain several occurrences of the same vertex. For example, a boundary cycle of the graph in Figure~\ref{example of an outerplanar graph} may be $(a,c,d,e,f,c,b)$.

Let $n_v$ be the number of connected components in $G[U]$ after the removal of the vertex $v$.
It is easy to see that $n_v$ is the number of occurrences of $v$ in any boundary cycle of $G[U]$.
The node $v$ is said to be a \emph{cut vertex} in $G[U]$ when $n_v\geq 2$.

\tikzstyle{vertex}  =[shape=circle , minimum size=20pt, draw=black]
\begin{figure}
\centering
\scalebox{0.7}{
\begin{tikzpicture}[node distance=50pt, auto]
    \node[vertex] (a){$a$};
    \node[vertex] (b) [below of = a] {$b$};
    \node[vertex] (c) [below right of = a]{$c$};
    \node[vertex] (d) [above right of = c]{$d$};
    \node[vertex] (e) [below right of = d]{$e$};
    \node[vertex] (f) [below right of = c]{$f$};
    \path   (a) edge (b)
            (a) edge (c)
            (b) edge (c)
            (c) edge (d)
            (d) edge (e)
            (e) edge (f)
            (f) edge (c);
\end{tikzpicture}
}
\caption{Example of an outerplanar graph}
\label{example of an outerplanar graph}
\end{figure}
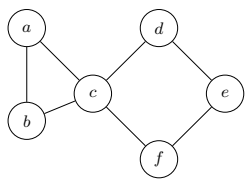

\label{appendix:algo-implement}
\subsection{Find a Shortest Path}
\label{sec:shortest path}
Let $U$ be a self-contained subset of $V$. Let $a$ and $b$ be two nodes in $U$.
The following algorithm \textsc{Shortest-Path}$(a,b,U)$ returns some shortest path between $a$ and $b$ using $O(|U|\log |U|)$ queries.
\vspace{3mm}
\begin{algorithm}{Shortest-Path}{a,b,U}
\begin{IF}{\delta(a,b)=1}
    \RETURN \text{path}(a,b)
\end{IF}\\
\CALL{Query}(U,a)\\
\CALL{Query}(U,b)\\
T\=\{u\in U\mid \delta(u,a)+\delta(u,b)=\delta(a,b)\}\\
c\=\text{any node in $T$ such that $\delta(c,a)=\lfloor\delta(a,b)/2\rfloor$}\\
U_1\=\{u\in T \mid \delta(u,a)<\delta(c,a)\}\\
U_2\=\{u\in T \mid \delta(u,a)>\delta(c,a)\}\\
P_1\=\CALL{Shortest-Path}(a,c,U_1)\\
P_2\=\CALL{Shortest-Path}(c,b,U_2)\\
\RETURN  \text{the concatenation of} P_1 \text{and} P_2
\end{algorithm}

First, we make $2|U|$ queries to get $\delta(u,a)$ and $\delta(u,b)$ for every $u\in U$. The set $T$ is the set of nodes in at least one shortest path between $a$ and $b$. Let $c$ be the node in the middle of some shortest path between $a$ and $b$. Then we construct recursively a shortest path between $a$ and $c$ and a shortest path between $c$ and $b$. The concatenation of these two paths is a shortest path between $a$ and $b$.

During the recursion, the distance between the two given endpoints is reduced to half at every level, so there are at most $O(\log |U|)$ levels of recursion. The number of queries at every level $O(|U|)$, since the sets in the same level of recursion are disjoint. So the total query complexity is $O(|U|\log |U|)$.

\subsection{Partition by Node}
\label{sec:partition by node}
Let $U$ be a self-contained subset of $V$ and let $x$ be a node in $U$. The set $U$ is separated into $n_x$ subsets $S_{x,1},\dots,S_{x,n_x}$ after the removal of $x$. Recall that $\{S_{x,1}^*,\dots,S_{x,n_x}^*\}$ is the partition of $U$ by $x$, where $S_{x,i}^*=S_{x,i}\cup\{x\}$.

The following algorithm \textsc{Partition-by-node}$(x,U)$ computes this partition using $O(\Delta\cdot|U|)$ queries.
In the algorithm, $Y$ is the set of neighbors of $x$. For any $y_1$ and $y_2$ in $Y$, we say that they are \emph{consecutive neighbors} if there exists a path between $y_1$ and $y_2$ that does not go by nodes in $Y\backslash\{y_1,y_2\}$. Two neighbors $y$ and $z$ are in the same $S_{x,i}$ (for some $i\in[1,n_x]$) if and only if there exists $v_1,\dots,v_k$ such that $v_1=y, v_k=z$, and every $(v_i,v_{i+1})_{1\leq i<k}$ is a pair of consecutive neighbors. The algorithm maintains a Disjoint-set data structure with $Union$ operation on consecutive neighbors. This leads to the partition $\{Y_1,\cdots Y_{n_x}\}$ of $Y$, which is the same as $\{S_{x,1}^*\cap Y,\dots,S_{x,n_x}^*\cap Y\}$.
We then classify  every $u\in U\backslash\{x\}$  into some $S_{x,i}^*$, according to which $Y_i$ contains the nodes nearest to $u$.

\vspace{3mm}
\begin{algorithm}{Partition-by-node}{x, U}
\CALL{Query}(U,x)\\
Y\=\{u\in U\mid \delta(x,u)=1\}\\
\CALL{Query}(U,Y)\\
\begin{FOR}{u\in U}
    d_u\=\min_{y\in Y}\{\delta(u,y)\}\\
    A_u\=\{y\in Y\mid \delta(u,y)=d_u\}\\
    \begin{IF}{|A_u|=2}
        (a_1,a_2)\=\text{the only two elements in $A_u$}\\
        Union(a_1,a_2)
    \end{IF}\\
    \begin{IF}{|A_u|=1}
        a\=\text{the only element in $A_u$}\\
        \begin{FOR}{y\in Y}
            \begin{IF}{\delta(u,y)=d_u+1}
                  Union(a,y)
            \end{IF}
        \end{FOR}
    \end{IF}
\end{FOR}\\
\{Y_1,\dots,Y_{n_x}\}\=\text{partition of $Y$ by the Disjoint-set data structure}\\
\begin{FOR}{i\=1 \TO n_x}
S_{x,i}^*\=\{u\in U\backslash\{x\}\mid A_u\subseteq Y_i\}
\end{FOR}\\
\RETURN \{S_{x,1}^*\cup \{x\}, \dots, S_{x,n_x}^*\cup \{x\}\}
\end{algorithm}

We will show that the partition $\{Y_1,\dots,Y_{n_x}\}$ in line 15 is indeed the partition $\{S_{x,1}^*\cap Y,\dots,S_{x,n_x}^*\cap Y\}$.
Consider any consecutive neighbors $y_1$ and $y_2$.
Let $v_1,\dots,v_l$ be a shortest path between $y_1$ and $y_2$ and let $u=v_{\lceil l/2\rceil}$.
When $l$ is odd, $\delta(u,y_1)=\delta(u,y_2)=d_v$, so $(y_1,y_2)$ is added to the Disjoint-set in line 9;
when $l$ is even, $\delta(u,y_1)=d_u$ and $\delta(u,y_2)=d_u+1$, so $(y_1,y_2)$ is added to the Disjoint-set in line 14.
Thus every pair of consecutive neighbors is added to the Disjoint-set.
For any $y\in Y$ and $z\in Y$ such that $y$ and $z$ are in the same $S_{x,i}^*$, there exists a path $v_1,\dots,v_k$ with $v_1=y$ and $v_k=z$, such that every $(v_i,v_{i+1})$ is a pair of consecutive neighbors and thus added to the Disjoint-set. So $y$ and $z$ are in the same partition.
On the other hand, it is easy to see that any node $y$ and node $z$ in $Union(y,z)$ are in the same $S_{x,i}^*$ for some $i\in [1,n_x]$.
So the partition of $Y$ in the algorithm is indeed $\{S_{x,1}^*\cap Y,\dots,S_{x,n_x}^*\cap Y\}$.

From this partition, we can construct $S_{x,i}^*$ easily (lines 16--17) by notifying that every $u\in S_{x,i}^*$ satisfies that $\delta(u,Y_i)<\delta(u,Y_j)$ for every $j\neq i$.
\subsection{Ordering of Neighbors}
\label{sec:Ordering of neighbors}
Let $U$ be a self-contained subset of $V$ and let $x$ be a non-cut vertex in $G[U]$. Let $Y$ be the set of neighbors of $x$ in $U$. Since $U$ is outerplanar, there exists an order $y_1,\dots,y_\lambda$ of the elements in $Y$ such that every $(y_i,y_{i+1})_{1\leq i<\lambda}$ is a pair of consecutive neighbors. Such order is unique up to symmetry, i.e., the only other order being $(y_\lambda,\dots,y_1)$. In fact, for every $y_i$ and $y_j$ under the above order with $|i-j|>1$, $y_i$ and $y_j$ are not consecutive neighbors, since otherwise $G[U]$ contains a subgraph homeomorphic from $K_4$, which contradicts the fact that $G[U]$ is outerplanar.

The algorithm \textsc{Neighbors-In-Order}$(x,U)$ makes $O(\Delta\cdot|U|)$ queries and returns the $\lambda$ neighbors $(y_1,\dots,y_\lambda)$ in order.

\vspace{3mm}
\begin{algorithm}{Neighbors-In-Order}{x, U}
\CALL{Query}(U,x)\\
Y\=\{u\in U\mid \delta(x,u)=1\}\\
\CALL{Query}(U,Y)\\
R=\emptyset\\
\begin{FOR}{u\in U}
    d_u\=\min_{y\in Y}\{\delta(u,y)\}\\
    A_u\=\{y\in Y\mid \delta(u,y)=d_u\}\\
    \begin{IF}{|A_u|=2}
        (a_1,a_2)\=\text{the only two elements in $A_u$}\\
        R\=R\cup \{(a_1,a_2)\}
    \end{IF}\\
    \begin{IF}{|A_u|=1}
        a\=\text{the only element in $A_u$}\\
        \begin{FOR}{y\in Y}
            \begin{IF}{\delta(u,y)=d_u+1}
                R\=R\cup \{(a,y)\}
            \end{IF}
        \end{FOR}
    \end{IF}
\end{FOR}\\
(y_1,\dots,y_\lambda)\=\text{an ordering of $Y$ s.t.\ for every  $1\leq i<\lambda$, $(y_i,y_{i+1})\in R$}\\
\RETURN (y_1,\dots,y_\lambda)
\end{algorithm}

The correctness of the algorithm follows from Proposition~\ref{prop:ordering neighbors}.
\begin{proposition}
\label{prop:ordering neighbors}
A pair $(u,v)$ is in $R$ iff. $u$ and $v$ are consecutive neighbors.
\end{proposition}

\begin{proof}
The same argument as in the last section shows that if $u$ and $v$ are consecutive neighbors, then $(u,v)\in R$.
So we only need to prove the other direction: during the for loop over every $u\in U$ (lines 5--15), the only pairs added to $R$ are consecutive neighbors.

For every $i\in[1,\lambda]$, let $\{S_{y_i,1},\dots,S_{y_i,n_{y_i}}\}$ be the partition of $U$ by $y_i$. Let $k_i\in [1,n_{y_i}]$ be such that $x\in S_{y_i,k_i}$ and let $V_i=U\backslash S_{y_i,k_i}^*$.
Define $T=S_{y_1,k_1}\cap\cdots\cap S_{y_{\lambda},k_\lambda}$.
Separate $T$ into $\lambda-1$ subsets: $T_1$,\dots,$T_{\lambda-1}$, such that $T_i$ contains nodes between $(x,y_i)$ and $(x,y_{i+1})$
(See Section~\ref{sec:partition by edge} for the definition of the partition by an edge).
Then $\{V_1,\dots,V_\lambda,T_1,\dots,T_{\lambda-1}\}$ is a partition of $U$ as in Figure~\ref{fig:partition by neighbors}~\footnote{The set $V_0$ in Figure~\ref{fig:partition by neighbors} does not exists here since $x$ is a non-cut node in $U$}.
It is sufficient to show that, during the For loop over every $u\in V_i$ ($1\leq i\leq \lambda$) and over every $u\in T_i$ ($1\leq i<\lambda$), the only pairs added to $R$ are consecutive neighbors.

Case 1:  $u$ is in some $V_i$ ($1\leq i\leq \lambda$).
Then $y_i$ is the only node in $Y$ with $\delta(u,y_i)=d_u$.
For every $y_j\in Y$ with $\delta(y_j,u)=d_u+1$, there must be an edge $(y_i,y_j)$ in the graph, thus $(y_i,y_j)$ are consecutive neighbors.

Case 2: $u$ is in some $T_i$ ($1\leq i<\lambda$).
If $\delta(u,y_i)=\delta(u,y_{i+1})=d_u$, we have $\delta(u,y_j)>d_u$ for every $y_j\in Y\backslash\{y_i,y_{i+1}\}$, so the pair $(y_i,y_{i+1})$ of consecutive neighbors is the only pair added to $R$.
Otherwise $\delta(u,y_i)\neq \delta(u,y_{i+1})$ and $\min\{\delta(x,y_i),\delta(u,y_{i+1})\}=d_u$.
Assume $\delta(u,y_i)=d_u$ without loss of generality. For every $y_j\in Y$ such that $\delta(u,y_j)=d_u+1$, either $y_j=y_{i+1}$ or there is an edge $(y_i,y_j)$ in the graph.
In both cases, $y_i$ and $y_j$ are consecutive neighbors.
\end{proof}

\subsection{Partition by Edge}
\label{sec:partition by edge}
Let $U$ be a self-contained subset of $V$. Let $x$ and $y$ be nodes in $U$ such that $(x,y)$ is an edge in $G$, thus also an edge in $G[U]$.
Let $\{S_{x,1},\dots,S_{x,n_x}\}$ be the partition of $U$ by the node $x$ and $\{S_{y,1},\dots,S_{y,n_y}\}$ be the partition of $U$ by the node $y$, which can be computed by the algorithm in Section~\ref{sec:partition by node} using $O(\Delta\cdot|U|)$ queries.
For every $S_{x,i}^*$ $(1\leq i\leq n_x)$ not containing $y$, remove from $U$ all its elements; and for every $S_{y,i}^*$ $(1\leq i\leq n_y)$  not containing $x$, remove from $U$ all its elements.
Now $x$ and $y$ are non-cut vertices in $G[U]$.

Consider any boundary cycle of $G[U]$. Both $x$ and $y$ appear exactly once in this cycle, so they separate it into two segments. Define $A^*$ and $B^*$ to be the sets of nodes in the two segments respectively (excluding $x$ and $y$). Let $A=A^*\cup\{x,y\}$ and $B=B^*\cup\{x,y\}$.
It is easy to see that every node in $V\backslash\{x,y\}$ belongs to either $A^*$ or $B^*$; and that for every $a\in A^*$ and $b\in B^*$, there is no edge between $a$ and $b$, since otherwise $G[U]$ contains a subgraph homeomorphic from $K_4$, which contradicts the fact that $G[U]$ is outerplanar.
The goal of this section is to compute the partition $\{A,B\}$.

Let $(z_1,\dots,z_\lambda)$ be the neighbors of $x$ in $U$ in order and let $(t_1,\dots,t_\mu)$ be the neighbors of $y$ in $U$ in order. These orders can be computed by the algorithm in Section~\ref{sec:Ordering of neighbors} using $O(\Delta\cdot|U|)$ queries. In any boundary cycle of $G[U]$, $(z_1,\dots,z_\lambda)$ is in the same order with either $(t_1,\dots,t_\mu)$ or $(t_\mu,\dots,t_1)$.
It is not hard to distinguish these two cases using $O(\Delta\cdot|U|)$ queries. In the following, we assume the first case holds without loss of generality. Let $i\in [1,\lambda]$ and $j\in[1,\mu]$ be such that $y=z_i$ and $x=t_j$. See Figure~\ref{fig:partition by edge}.

\tikzstyle{vertex}=[font=\small]

\begin{figure}
\centering
\scalebox{1}{
\begin{tikzpicture}
\node[vertex] (a1) at (90:2.5) {$x$};
\node[vertex] (a2) at (65:2.5) {$z_1$};
\node[vertex] (a3) at (40:2.5) {};
\node[vertex] (a4) at (15:2.5) {$z_{i-1}$};
\node[vertex] (a5) at (345:2.5) {$t_{j+1}$};
\node[vertex] (a6) at (320:2.5) {};
\node[vertex] (a7) at (295:2.5) {$t_\mu$};
\node[vertex] (a8) at (270:2.5) {$y$};
\node[vertex] (a9) at (245:2.5) {$t_1$};
\node[vertex] (a10) at (220:2.5) {};
\node[vertex] (a11) at (195:2.5) {$t_{j-1}$};
\node[vertex] (a12) at (165:2.5) {$z_{i+1}$};
\node[vertex] (a13) at (140:2.5) {};
\node[vertex] (a14) at (115:2.5) {$z_\lambda$};

\foreach \from/\to in {a2/a3,a3/a4/,a4/a5,a5/a6,a6/a7,a9/a10,a10/a11,a11/a12,a12/a13,a13/a14}
    \path (\from) edge [dashed, bend left=60,looseness=1] (\to);
\foreach \to in {a2,a3,a4,a14,a13,a12}
    \path (a1) edge (\to);
\foreach \to in {a5,a6,a7,a11,a10,a9}
    \path (a8) edge (\to);
\path (a8) edge (a1);

\draw [rounded corners=10,dashed] (a2) -- (60:3.5) -- (70:3.5) -- (a2);
\draw [rounded corners=10,dashed] (a3) -- (35:3.5) -- (45:3.5) -- (a3);
\draw [rounded corners=10,dashed] (a4) -- (10:3.5) -- (20:3.5) -- (a4);
\draw [rounded corners=10,dashed] (a5) -- (340:3.5) -- (350:3.5) -- (a5);
\draw [rounded corners=10,dashed] (a6) -- (315:3.5) -- (325:3.5) -- (a6);
\draw [rounded corners=10,dashed] (a7) -- (290:3.5) -- (300:3.5) -- (a7);

\draw [rounded corners=10,dashed] (a9) -- (240:3.5) -- (250:3.5) -- (a9);
\draw [rounded corners=10,dashed] (a10) -- (215:3.5) -- (225:3.5) -- (a10);
\draw [rounded corners=10,dashed] (a11) -- (190:3.5) -- (200:3.5) -- (a11);
\draw [rounded corners=10,dashed] (a12) -- (160:3.5) -- (170:3.5) -- (a12);
\draw [rounded corners=10,dashed] (a13) -- (135:3.5) -- (145:3.5) -- (a13);
\draw [rounded corners=10,dashed] (a14) -- (110:3.5) -- (120:3.5) -- (a14);
\end{tikzpicture}
}
\caption{Partition by the edge $(x,y)$}
\label{fig:partition by edge}
\end{figure}
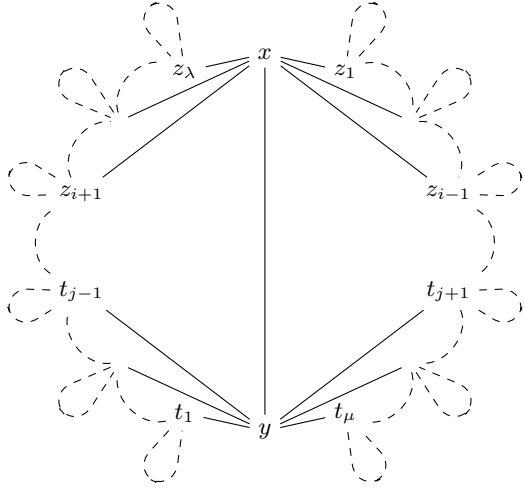

For any node $u\in U\backslash\{x,y\}$, the following algorithm computes whether $u$ is to the left or to the right with respect to the edge $(x,y)$ using a constant number of queries. Thus we obtain the partition $\{A,B\}$ using $O(\Delta\cdot|U|)$ queries.

\vspace{3mm}
\begin{algorithm}{Left-or-Right}{x,y,u}
\CALL{Query}(x,u)\\
\CALL{Query}(y,u)\\
\begin{IF}{\delta(x,u)\leq \delta(y,u)}
    \begin{FOR}{k\=1\TO\lambda}
        \CALL{Query}(z_k,u)
    \end{FOR}\\
    \text{Let $i^*$ be such that $\delta(z_{i^*},u)=\min_{1\leq k\leq \lambda}\{\delta(z_k,u)\}$}\\
    \begin{IF}{i^*<i}
        \RETURN \text{Right}
    \ELSE
        \RETURN \text{Left}
    \end{IF}
\ELSE
    \begin{FOR}{k\=1\TO\mu}
        \CALL{Query}(t_k,u)
    \end{FOR}\\
    \text{Let $j^*$ be such that $\delta(t_{j^*},u)=\min_{1\leq k\leq \mu}\{\delta(t_k,u)\}$}\\
    \begin{IF}{j^*<j}
        \RETURN \text{Left}
    \ELSE
        \RETURN \text{Right}
    \end{IF}
\end{IF}
\end{algorithm}

We will show that the algorithm \textsc{Left-or-Right}$(x,y,u)$ returns the correct side of $u$ with respect to $(x,y)$.
First consider the case that $\delta(x,u)\leq \delta(y,u)$.
Let $i^*\in [1,\lambda]$ be such that $z_{i^*}$ is the closest to $u$ among all neighbors of $x$.
The node $z_{i^*}$ is different from $y$ since we assume that $\delta(x,u)\leq \delta(y,u)$.
If $i^*<i$, then $z_{i^*}$ is to the right of $(x,y)$, so is $u$;
if $i^*>i$, then $z_{i^*}$ is to the left of $(x,y)$, so is $u$.
Thus the algorithm above returns the correct side of $u$.
The case that $\delta(x,u)>\delta(y,u)$ is similar.

\subsection{Finding Polygon}
\label{sec:finding polygon}
Let $U$ be a self-contained subgraph of $V$. Let $x,y_i,y_{i+1}$ be nodes in $U$ such that $x$ is a non-cut node in $G[U]$, $y_i$ and $y_{i+1}$ are neighbors of $x$ and they are consecutive (see Section~\ref{sec:partition by node} for the definition of \emph{consecutive neighbors}). There exists some path in $G[U]$ between $y_i$ and $y_{i+1}$ without $x$, since $x$ is a non-cut node in $G[U]$. Let $P$ be such a path of minimum length. Then $L$ is unique, since otherwise $G[U]$ contains a subgraph homeomorphic from $K_{2,3}$, which contradicts with the fact that $G[U]$ is outerplanar. The path $P$ and the edges $(x,y_i)$, $(x,y_{i+1})$ together form a polygon, which is the unique polygon using $(x,y_i)$ and $(x,y_{i+1})$ as edges. The following algorithm computes this polygon using $O(\Delta\cdot|U|\log|U|)$ queries.

\vspace{3mm}
\begin{algorithm}{Find-Polygon}{x,y_i,y_{i+1},U}
A_1\= \text{the subset of $U$ separated by the edge $(x,y_i)$ which contains $y_{i+1}$}\\
A_2\= \text{the subset of $U$ separated by the edge $(x,y_{i+1})$ which contains $y_i$}\\
A\=A_1\cap A_2\\
\CALL{Query}(A,y_i)\\
\CALL{Query}(A,y_{i+1})\\
d\=\min_{u\in A}\{\delta(u,y_i)+\delta(u,y_{i+1})\}\\
\text{Let $z\in A$ be such that $\delta(z,y_i)+\delta(z,y_{i+1})=d$ and $\delta(z,y_i)=\lfloor d/2\rfloor$}\\
P_1\=\CALL{Shortest-Path}(y_i,z)\\
P_2\=\CALL{Shortest-Path}(z,y_{i+1})\\
\RETURN \text{the concatenation of } P_1, P_2, (y_{i+1},x), (x,y_i)
\end{algorithm}

The set $A$ above is the set of nodes between the edges $(x,y_i)$ and $(x,y_{i+1})$.
Let $d$ be the length of the shortest path $P$ between $y_i$ and $y_{i+1}$ that does not go by $x$.
Let $z$ is the node in the middle of $P$.
We then calculate the shortest path $P_1$ between $y_i$ and $z$ and the shortest path $P_2$ between $z$ and $y_{i+1}$ using $O(|U|\log|U|)$ queries (see Section~\ref{sec:shortest path}).
The concatenation of $P_1$ and $P_2$ is $P$. Together with the edges $(y_{i+1},x)$ and $(x,y_i)$, we get the polygon.
\subsection{Partition by Polygon}
\label{sec:partition by polygon}
Let $U$ be a self-contained subset of $V$.
Let $(q_1,\dots,q_l)$ be an arbitrary polygon where every $q_i$ is in $U$.
For every $i\in[1,l]$, let $\{S_{q_i,j}\}_{1\leq j\leq n_{q_i}}$ be the partition of $U$ by the node $q_i$. Exactly one of these $n_{q_i}$ subsets contains all nodes of the polygon. Let it be $S_{q_i,m_i}$ for some $m_i\in [1,n_{q_i}]$ and define $W_i=U\backslash S_{q_i,m_i}$ (see Figure~\ref{fig:partition by polygon}). Let $R=\bigcap_{1\leq i\leq l} S_{q_i,m_i}$. Since $R$ is self-contained, $R$ can be partitioned into $l$ subsets $R_1,\dots,R_l$ by the $l$ edges of the polygon. The goal of this section is to obtain the partition: $\{W_1,\dots,W_l,R_1,\dots,R_l\}$.
Section~\ref{sec:partition by node} and Section~\ref{sec:partition by edge} give algorithms to compute any $W_i$ and any $R_i$ using $O(\Delta\cdot|U|)$ queries.
So a naive algorithm to compute the partition has query complexity $O(l\cdot\Delta\cdot|U|)$, which is $O(\Delta\cdot|U|^2)$ when $l=\Theta(|U|)$.
Next we will give an improved implementation that uses $O(\Delta\cdot|U|\log |U|)$ queries based on dichotomy.

First we compute the eight subsets $W_1$, $R_1$, $W_{\lfloor l/2\rfloor}$, $R_{\lfloor l/2\rfloor}$, $W_{\lfloor l/2\rfloor+1}$, $R_{\lfloor l/2\rfloor+1}$, $W_l$, $R_l$. For every node $v\in U$ outside the eight subsets, it belongs to one of the two subsets $Z_1=\bigcup_{2\leq i\leq \lfloor l/2\rfloor-1} (W_i\cup R_i)$ and $Z_2=\bigcup_{\lfloor l/2\rfloor+2\leq i\leq l-1} (W_i\cup R_i)$. It is easy to see that both subsets are self-contained. In addition, we can decide if $v$ belongs to $Z_1$ or $Z_2$ by making two queries: \textsc{Query}$(v,q_2)$ and \textsc{Query}$(v,q_l)$. If $\delta(v,q_2)<\delta(v,q_l)$, then $v$ is in $Z_1$, otherwise $v$ is in $Z_2$.
Thus we obtain the sets $Z_1$ and $Z_2$ using $O(|U|)$ queries.

The following algorithm \textsc{Partition-Segment}$(s,t,Z)$ receives two integers $s,t\in [1,l]$ and a self-contained subset $Z\subseteq U$, such that $Z=\bigcup_{s\leq i\leq t}(W_i\cup R_i)$, and returns a partition $\{W_s,R_s,\dots,W_t,R_t\}$ of $Z$.

\vspace{3mm}
\begin{algorithm}{Partition-Segment}{s,t,Z}
\begin{IF}{s>t}
\RETURN \emptyset
\end{IF}\\
m=\lfloor(s+t)/2\rfloor\\
\text{Compute} W_m,R_m\\
A\=(Z\backslash(W_m\cup R_m))\cup \{q_m,q_{m+1}\}\\
\CALL{Query}(A,q_m)\\
\CALL{Query}(A,q_{m+1})\\
A_1\=\{u\in A\mid q_m<q_{m+1}\}\\
A_2\=\{u\in A\mid q_m>q_{m+1}\}\\
Q_1\=\CALL{Partition-Segment}(s,m-1,A_1)\\
Q_2\=\CALL{Partition-Segment}(m+1,t,A_2)\\
\RETURN Q_1\cup Q_2\cup \{W_m,R_m\}
\end{algorithm}

The number of queries to compute $W_m$ and $R_m$ is $O(\Delta\cdot|Z|)$. During the recursion, every time  $(t-s)$ is reduced to a half, so there are at most $\log l\leq \log |U|$ levels of recursion. At every level, the query complexity is $O(\Delta\cdot|U|)$, since the sets $Z$ in the same level are all disjoint. So the total query complexity of this algorithm is $O(\Delta\cdot|U|\log |U|)$.

Let $Q_1$ (resp. $Q_2$) be the partition of $Z_1$ (resp. $Z_2$) returned by the algorithm above. Then $Q_1\cup Q_2$ together with the eight sets computed at the beginning give the partition $\{W_1,R_1,\dots,W_l,R_l\}$.

\section{Missing proofs for Balanced-Partition Algorithm}\label{appendix:proofs}
\subsection{Proof of Lemma~\ref{lemma:self-contained} }\label{appendix:self-contained}

We only need to show that every time when the algorithm  fails to provide a $\beta$-balanced partition for a sampling (by executing \emph{go to Step 1}) in Step 3, 4, 7 or 8, there must be $\beta$-bad set.

In Step 3, this happens when $x$ has exactly one neighbor in $D$. Thus $D\backslash\{x\}$ is a $\beta$-bad set.
In Step 4, this happens when $|V_i|>\beta|U|$ for some $i\in[1,\lambda]$. Such $V_i$ is a $\beta$-bad set since $x\notin V_i$ for every $i\in[1,\lambda]$.
In Step 7, this happens when $|W_i|>\beta|U|$ for some $i\in[1,l]$. Such $W_i$ cannot be $W_1$, because $W_1$ is the same as $V_0$, which has size $|U|-|D|+1< |U|-\beta|U|+1<\beta|U|$, since $\beta\in (0.7,1)$ and $|U|\geq 10$. Notice that $x\notin W_i$ for every $i\in[2,l]$, so any $W_i$ with $|W_i|>\beta|U|$ is a $\beta$-bad set.
In Step 8, this happens when $|R_i|>\beta|U|$ for some $i\in[1,l]$. Such $R_i$ cannot be $R_1$ or $R_l$. In fact, $|R_1|+|R_l|\leq |U|-|T_j|+4< |U|-\beta |U|+4\leq \beta|U|$ for $\beta\in(0.7,1)$ and $|U|\geq 10$. Since $x\notin R_i$ for every $i\in[2,l-1]$, any $R_i$ with $|R_i|>\beta|U|$ is a $\beta$-bad set.

\subsection{Proof of Lemma~\ref{most popular node} }\label{appendix:most-popular-node}

We first prove the following lemma:
\begin{lemma}
\label{tree-separator}
In any tree of bounded degree $\Delta$ which is not a singleton, there is an edge that separates the tree into two parts, such that both parts contain at least $\frac{1}{2\Delta}$ fraction of nodes.
\end{lemma}

\begin{proof}
Let $T$ be a degree bounded tree of size $n$ (for any $n\geq 2$).
For any edge $(u,v)$ in $T$, the removal of this edge would separate $T$ into two subtrees. Let $A_{u,v}$ (resp. $B_{u,v}$) be the set of nodes in the subtree containing $u$ (resp. $v$).
Let $(u^*,v^*)$ be the edge which maximizes $\min(|A_{u,v}|,|B_{u,v}|)$. We will show that both $|A_{u^*,v^*}|$ and $|B_{u^*,v^*}|$ are at least $\frac{1}{2\delta}\cdot n$.

Without loss of generality, we assume that $|A_{u^*,v^*}|\geq |B_{u^*,v^*}|$. Let $w_1,\dots,w_\lambda$ be neighbors of $u^*$ which are different from $v^*$ (where $\lambda\leq \Delta-1$).
Then $|B_{w_i,u^*}|> |B_{u^*,v^*}|$ for every $i\in [1,\lambda]$, since $B_{u^*,v^*}$ is a strict subset of every $B_{w_i,u^*}$.
On the other hand, $\min(|A_{w_i,u^*}|,|B_{w_i,u^*}|)\leq \min(|A_{u^*,v^*}|,|B_{u^*,v^*}|)=|B_{u^*,v^*}|$, where the inequality is from the definition of $(u^*,v^*)$. So $|A_{w_i,u^*}|\leq |B_{u^*,v^*}|$ for every $i\in [1,\lambda]$. Since $A_{u^*,v^*}=\{u\}\uplus A_{w_1,u^*}\uplus\cdots\uplus A_{w_\lambda,u^*}$, we have:
$|A_{u^*,v^*}|=1+(|A_{w_1,u^*}|+\cdots+|A_{w_\lambda,u^*}|)$, which is at most $1+(\Delta-1)|B_{u^*,v
^*}|$.
Since $|A_{u^*,v^*}|+|B_{u^*,v^*}|=n$, we then have $|B_{u^*,v^*}|\geq \frac{n-1}{\Delta} \geq \frac{n}{2\Delta}$ for any $n\geq 2$.
\end{proof}

Let $G=(V,E)$ be any outerplanar graph of bounded degree with $n$ vertices.
When $G$ is a singleton, the result is trivial. So we assume that $n\geq 2$.
When $G$ is a tree, take $z$ to be one of the endpoints of the edge satisfying the condition of Lemma~\ref{tree-separator}. We have $p_z\geq\alpha_0$ for $\alpha_0=\frac{1}{2\Delta}\left(1-\frac{1}{2\Delta}\right)$.

Next we consider the case when $G$ is not a tree.
Take $(q_1,\dots,q_l)$ to be a cycle in $G$ of minimum length. This cycle must form a polygon, since otherwise we can separate it into two smaller cycles. We partition $U$ into $2l$ self-contained subsets: $R_1,\dots,R_l,W_1,\dots,W_l$ (as in Section~\ref{sec:partition by polygon}).
At least one of the four cases holds:

\begin{enumerate}
\item Every subset is of size smaller than $\frac{1}{100}n$;
\item There exists $i\in [1,l]$ such that $|W_i|\in[\frac{1}{100}n, \frac{99}{100}n]$;
\item There exists $j\in [1,l]$ such that $|R_j|\in[\frac{1}{100}n, \frac{99}{100}n]$;
\item Exactly one subset is of size larger than $\frac{99}{100}n$.
\end{enumerate}

In Case 1, since every $(R_i\cup W_i)_{1\leq i\leq l}$ has at most $\frac{1}{50}n$ nodes, there exists $i\in [1,l]$ such that the two subsets $A=\bigcup_{1< j< i} R_j\cup W_j$ and $B=\bigcup_{i< j\leq l} R_j\cup W_j$ both have more than $(\frac{1}{2}-\frac{1}{25})n$ nodes. For every $a\in A$ and every $b\in B$, the shortest path between $a$ and $b$ must go by either $q_1$ or $q_i$. Thus $\max\{p_{q_1},p_{q_i}\}\geq \alpha_1$, where $\alpha_1=(\frac{1}{2}-\frac{1}{25})^2$.

In Case 2, let $i\in [1,l]$ be such that $|W_i|\in[\frac{1}{100}n, \frac{99}{100}n]$.
For every pair $(a,b)\in W_i\times(U\backslash W_i)$, the node $q_i$ is in every shortest path between $a$ and $b$.
So $p_{q_i}\geq \alpha_2$, where $\alpha_2=2\cdot\frac{1}{100}\cdot\frac{99}{100}$.

In Case 3, let $j\in [1,l]$ be such that $|R_j|\in[\frac{1}{100}n, \frac{99}{100}n]$.
For every pair $(a,b)\in R_j\times(U\backslash R_j)$, every shortest path between $a$ and $b$ goes by at least one of $q_j$ and $q_{j+1}$, thus $\max\{p_{q_j},p_{q_{j+1}}\}\geq \alpha_3$, where $\alpha_3=\frac{1}{100}\cdot\frac{99}{100}$.

In Case 4, let $T_0\subset U$ be the set of size larger than $\frac{99}{100}n$. We already know that $T_0$ is self-contained. If $G[T_0]$ is a tree, by  Lemma~\ref{tree-separator}, there exists $z\in T_0$ such that $z$ is in at least $\alpha_0$ fraction of the shortest paths where both endpoints are in $T_0$. Since $|T_0|>\frac{99}{100}n$, $z$ is in at least $\left(\frac{99}{100}\right)^2\cdot\alpha_0$ fraction of the shortest paths where both endpoints are in $U$.
If $G[T_0]$ is a tree, it contains a polygon. If $G[T_0]$ is in Case 1 (resp. Case 2, Case 3), similarly, there exists $z\in T_0$ such that $p_z$ is at least $\left(\frac{99}{100}\right)^2\cdot\alpha_1$ (resp. $\left(\frac{99}{100}\right)^2\cdot\alpha_2$, $\left(\frac{99}{100}\right)^2\cdot\alpha_3$).
We only need to treat the Case 4 of $G[T_0]$. Let $T_1\subset T_0$ be the set of size larger than $\frac{99}{100}n$. We apply the same argument on $G[T_1]$. If $G[T_1]$ is not in Case 4, we are done; otherwise we obtain $T_2$ with $|T_2|>\frac{99}{100}n$, etc. Every $T_{i+1}$ is a strict subset of $T_i$ (for $i\geq 0$). So this procedure stops after a finite number of iterations and finds a node $z$ with $p_z\geq\alpha$ where $\alpha=\left(\frac{99}{100}\right)^2\cdot\min\{\alpha_0,\alpha_1,\alpha_2,\alpha_3\}$.

\subsection{Proof of Lemma~\ref{sample accuracy} }
\label{appendix:sample-accuracy}

Let $z$ be a node in $U$ with $p_z\geq\alpha$ (such node always exists by Lemma~\ref{most popular node}).
We will show that $\Proba{\widehat{p}_y>\widehat{p}_z}<\frac{1}{3|U|}$ for any node $y$ with $p_y\leq \alpha/2$. This is sufficient since we then have  $\Proba{\exists y\in U, \text{ s.t.\ } p_y\leq \alpha/2 \text{ and }\widehat{p}_y>\widehat{p}_z}<\frac{1}{3}$. Thus with probability at least $\frac{2}{3}$, any node $x$ with the largest $\widehat{p}_x$ satisfies $p_x>\alpha/2$.

Let $y$ be any node with $p_y\leq \alpha/2$.
For every $i\in [1,\omega]$, define the variable $Y_i\in \{0,1\}$, such that $Y_i=1$ if the node $y$ is in some shortest path between $a_i$ and $b_i$, and $Y_i=0$ otherwise. Since $\{a_i\}_{1\leq i\leq \omega}$ and $\{b_i\}_{1\leq i\leq \omega}$ are uniform and independent random nodes in $U$, $\{Y_i\}_{1\leq i\leq \omega}$ are independent and identically distributed random variables, and each $Y_i$ equals $1$ with probability $p_y$. We then have $\mathbb{E}[Y_i]=p_y\leq \alpha/2$.

Similarly, define $Z_i\in \{0,1\}$ such that $Z_i=1$ if the node $z$ is in some shortest path between $a_i$ and $b_i$, and $Z_i=0$ otherwise. Then $\{Z_i\}_{1\leq i\leq \omega}$ are independent and identically distributed random variables and each $Z_i$ equals $1$ with probability $p_z$. Thus $\mathbb{E}[Z_i]=p_z\geq\alpha$.

For every $i\in [1,\omega]$, define $T_i=Y_i-Z_i$. Then $\mathbb{E}[T_i]=\mathbb{E}[Y_i]-\mathbb{E}[Z_i]<-\frac{\alpha}{2}$. Let $\overline{T}=\frac{1}{t}(T_1+\cdots T_\omega)$. Since $\{T_i\}_{1\leq i\leq \omega}$ are independent and identically distributed random variables, $\mathbb{E}[\overline{T}]=\mathbb{E}[T_1]$. We have:

\[
\Proba{\widehat{p}_y>\widehat{p}_z}  = \Proba{\sum_{1\leq i\leq \omega}T_i >0}
\leq  \Proba{|\overline{T}-\mathbb{E}[\overline{T}]|>\frac{\alpha}{2}}
\leq 2\cdot\exp(-\frac{\alpha^2\omega}{8})
\]
where the last step holds by Hoeffding's inequality.
Take the constant $C$ to be large enough such that $2\cdot\exp(-\frac{\alpha^2\omega}{8})<\frac{1}{3|U|}$ as soon as $|U|\geq 1$, where $\omega=C\cdot\log |U|$. Then we have $\Proba{\widehat{p}_y>\widehat{p}_z}<\frac{1}{3|U|}$.

\section{Proof of Proposition~\ref{prop:approx-lower-bound}}\label{appendix:approx-lower-bound}

We will define a certain distribution of graphs and show that on that random input, any {deterministic} algorithm for the $f$-approximate metric reconstruction problem requires $\Omega(n^2/f)$ queries on average, for $n$ large enough. By Yao's Minimax Principle~\cite{yao1977probabilistic}, the result follows.

To simplify the proof, let $n=2f k+1$, for $k\in\mathbb{N}$. We take the uniform distribution on the following set of trees, one for each $f$-tuple $(\sigma_1,\dots,\sigma_f)$, where every $\sigma_i$ is a permutation of $S_k$. The tree $T$ has one vertex $a_0$ as the root (on the first level), $k$ vertices $a_1,\dots, a_k$ on the second level, $k$ vertices $a_{k+1},\dots,a_{2k}$ on the third level, $\cdots$, and $k$ vertices $a_{n-k},\dots,a_{n-1}$ on the $(2f+1)^{\rm th}$ level.
For every $l\in [2,f]$ and every $i\in [1,k]$, there is an edge between the $i^{\rm th}$ node on level $l$ and the $i^{\rm th}$ node on level $l+1$.
For every $l\in [f+1,2f]$ and every $i\in [1,k]$, there is an edge between the $i^{\rm th}$ node on level $l$ and the $\sigma_{l-f}(i)^{\rm th}$ node on level $l+1$.
Every tree constructed above has $k$ branches from the root, and every branch is a path of $2f$ nodes.
We will show that any deterministic algorithm requires $\Omega(n^2/f)$ queries on average to compute an $f$-approximation of the metric of $T$, for $n$ large enough.

Let $\mathcal{A}$ be a deterministic algorithm for an $f$-approximation of the metric.
Based on $\mathcal{A}$, we can reconstruct the tree \emph{exactly} as follows:
\begin{enumerate}
\item Execute the algorithm $\mathcal{A}$ to get $\widehat{\delta}$ as the $f$-approximation of the metric;
\item For every $u$ and $v$ on consecutive levels below level $l+1$, there is an edge between $u$ and $v$ iff. $\widehat{\delta}(u,v)<2f$.
\end{enumerate}

In fact, for every two nodes $u$ and $v$ on consecutive levels below level $f+1$, if they are in the same branch with respect to the root, we have $\delta(u,v)=1$, thus $\widehat{\delta}(u,v)\leq f$; and if they are in different branches, we have $\delta(u,v)=\delta(a_0,u)+\delta(a_0,v)\geq 2f$, thus $\widehat{\delta}(u,v)\geq 2f$. So we indeed reconstruct the tree $T$ exactly based on $A$.
In order to prove that $\mathcal{A}$ requires $\Omega(n^2/f)$ queries on average, we only need to prove that any deterministic algorithm for the exact reconstruction problem requires $\Omega(n^2/f)$ queries on average.

Let $\mathcal{A'}$ be a deterministic algorithm that reconstructs exactly the tree $T$. We assume that $\mathcal{A'}$ does not make redundant queries, whose answers can be  deduced before the query. Obviously, any query with the root is redundant. For any two node $u$ and $v$, let $l_u$ and $l_v$ be their levels. The query $(u,v)$ is redundant when $l_u\leq f+1$ and $l_v\leq f+1$, since the first $f+1$ levels in $T$ are fixed. Thus every query $(u,v)$ is such that $l_u>f+1$ and $l_v\geq 2$ (we suppose that $l_u\geq l_v$ without loss of generality). The answer is either $l_u-l_v$, if $u$ and $v$ are in the same branch; or $l_u+l_v-2$, if $u$ and $v$ are in different branches. We can equivalently identify the answer as \emph{Yes} or \emph{No} to the question: \emph{Are $u$ and $v$ in the same branch? }

Next, we will bound the number of \emph{Yes} answers received by $\mathcal{A'}$. To do this, we introduce the \emph{component graph} $H$, which represents the information from all \emph{Yes} answers received by $\mathcal{A'}$. The vertex set of $H$ is defined to be the set of all nodes in $T$ on level at least $f+1$. At the beginning, the edge set of $H$ is empty. Every time when a query $(u,v)$ gives a \emph{Yes} answer, we add an edge to $H$:
\begin{itemize}
\item if $l_u> f+1$ and $l_v\geq f+1$, add the edge $(u,v)$ to $H$;
\item if $l_v> f+1$ and $2\leq l_v< f+1$, let $w$ be the only node on level $f+1$ which is in the same branch of $v$; add the edge $(u,w)$ to $H$;
\end{itemize}

There could not be cycles in $H$, since otherwise there are redundant queries. The number of connected components in $H$ is at least $k$, since every connected component in $H$ contains nodes from the same branch of $T$ and there are $k$ branches in $T$. The number of edges in $H$ is the number of vertices in $H$ minus the number of connected components in $H$, so it is at most $k(f+1)-k=k f$. There is a one-to-one correspondence between the \emph{Yes} answers and the edges in $H$, so $\mathcal{A'}$ stops after at most $kf$ \emph{Yes} answers.

We use a decision tree argument. Consider the decision tree of $\mathcal{A'}$ (see Figure~\ref{decision tree}). $\mathcal{A'}$ first queries some pair $(u_1,v_1)$. If the answer is \emph{Yes} (left subtree in Figure~\ref{decision tree}), it queries some pair $(u_2,v_2)$, otherwise (right subtree in Figure~\ref{decision tree}) it queries some pair $(u_3,v_3)$, etc. Since $\mathcal{A'}$ is deterministic, the sequence $\{ (u_i, v_i)\}_{i\geq 1}$ is fixed in advance.

\tikzstyle{level 1}=[level distance=2cm, sibling distance=8cm]
\tikzstyle{level 2}=[level distance=2cm, sibling distance=4cm]
\tikzstyle{level 3}=[level distance=2cm, sibling distance=2cm]
\tikzstyle{level 3}=[level distance=3cm, sibling distance=1cm]
\tikzstyle{query} = [draw]

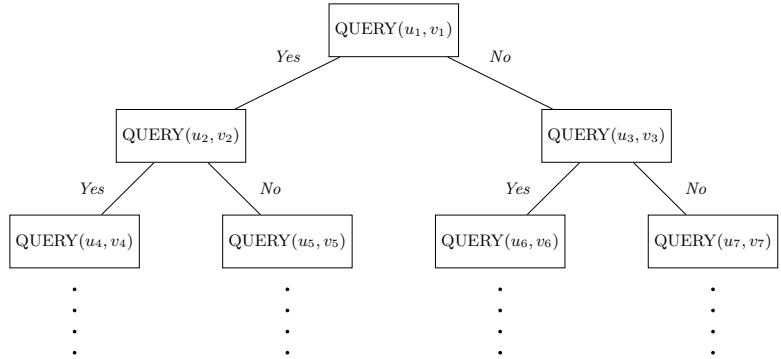
\begin{figure}
\centering
\scalebox{0.7}{
\begin{tikzpicture}[minimum size = 1cm]
\tikzset{decorate sep/.style 2 args=
{decorate,decoration={shape backgrounds,shape=circle,shape size=#1,shape sep=#2}}}
\node [query] (a){$\QUERY(u_1,v_1)$}
  child {
    node [query] (b) {$\QUERY(u_2,v_2)$}
    child {
      node [query] (c) {$\QUERY(u_4,v_4)$}
      child {
        node [] (d) {}
        edge from parent [draw=none]
      }
      edge from parent node [left=5pt] {\emph{Yes}}
    }
    child {
      node [query] (f) {$\QUERY(u_5,v_5)$}
      child {
        node (g) {}
        edge from parent [draw=none]
      }
      edge from parent node [right=5pt] {\emph{No}}
    }
    edge from parent node [above=0pt] {\emph{Yes}}
  }
  child {
    node [query] (i) {$\QUERY(u_3,v_3)$}
    child {
      node [query] (j) {$\QUERY(u_6,v_6)$}
      child {
        node (k) {}
        edge from parent [draw=none]
      }
      edge from parent node [left=5pt] {\emph{Yes}}
    }
    child {
      node [query] (m) {$\QUERY(u_7,v_7)$}
      child {
        node  (n) {}
        edge from parent [draw=none]
      }
      edge from parent node [right=5pt] {\emph{No}}
    }
    edge from parent node [above=0pt] {\emph{No}}
  };

  \draw[decorate sep={0.5mm}{4mm},fill] ($(c.south) + (0,-4mm)$) -- (d);
  \draw[decorate sep={0.5mm}{4mm},fill] ($(f.south) + (0,-4mm)$) -- (g);
  \draw[decorate sep={0.5mm}{4mm},fill] ($(j.south) + (0,-4mm)$) -- (k);
  \draw[decorate sep={0.5mm}{4mm},fill] ($(m.south) + (0,-4mm)$) -- (n);
\end{tikzpicture}
}
\vspace{-0.5cm}
\caption{Decision tree of $\mathcal{A'}$}
\label{decision tree}
\end{figure}

$\mathcal{A'}$ stops when it is able to reconstruct the tree $T$, which corresponds to a leaf in the decision tree. There are in total $(k!)^f$ leaves in the decision tree, corresponding to all $f$-tuples of permutations of $S_k$. For every leaf, there are $j\leq k f$ \emph{Yes} answers along the path from the root to this leaf.

Let $h>4k f$ to be defined later. The expected number of queries is
\[E [\# \hbox{queries} ] \geq h \Pr [ \# \hbox{queries}\geq h ] = h \left(1- \frac{\#  (\hbox{ leaves at level }<h ) } {(k!)^f}\right) .\]

A leaf of level $<h$ is identified by its root-leaf path, a word over $\{  \emph{Yes} , \emph{No}\}$ of length $< h$ and with $j\leq k f$ \emph{Yes}'s. Thus, using $k f\leq (1/2) (h/2)$:
\[ \#  (\hbox{ leaves at level }<h )  \leq \displaystyle\sum_{0\leq j \leq k f}\binom{h}{j} \leq 2\cdot \binom{h}{k f}\leq \frac{2h^{k f}}{(k f)!}. \]
Plugging that in, using Stirling's formula, and setting $h=k^2 f/e^2$, we get
$$E [\# \hbox{queries} ] \geq h\left(1- \frac{2}{(2\pi k f)^{\frac{1}{2}}(2\pi k)^{\frac{f}{2}}}\left( \frac{  he^2 } {k^{2} f}\right)^{k f} \right) \geq \frac{k^2 f}{e^2} (1-o(1)).$$
Since $k\cdot f\sim n/2$ and $e^2<10$, for $n$ large enough the rightmost expression is at least $n^2/(40f)=\Omega(n^2/f)$.

\fi

\end{document}